\documentclass[11pt]{article}
\usepackage[margin=3cm]{geometry}
\usepackage{amsmath,amssymb,amsthm}
\usepackage{graphicx}
\usepackage{tikz}
\usetikzlibrary{arrows.meta}
\usepackage[hidelinks,linktocpage=true]{hyperref}
\hypersetup{
 pdftitle={A Nearly Quadratic Lower Bound for Linear Optimization over Convex Bodies in the Membership Oracle Model},
 pdfauthor={Santosh S. Vempala}
}

\newcommand{\R}{\mathbb R}
\newcommand{\E}{\mathbb E}
\newcommand{\Prob}{\mathbb P}
\newcommand{\Var}{\operatorname{Var}}
\newcommand{\tr}{\operatorname{tr}}
\newcommand{\dist}{\operatorname{dist}}

\newcommand{\OPT}{\operatorname{OPT}}
\newcommand{\distp}{d_{\mathbb P}}

\theoremstyle{plain}
\newtheorem{theorem}{Theorem}[section]
\newtheorem{lemma}[theorem]{Lemma}
\newtheorem{corollary}[theorem]{Corollary}

\theoremstyle{definition}
\newtheorem{definition}[theorem]{Definition}
\theoremstyle{remark}

\title{A Nearly Quadratic Lower Bound for Linear Optimization\\ over Convex Bodies in the Membership Oracle Model}
\author{Santosh S. Vempala\thanks{Supported in part by NSF award CCF-2504994 and a Simons Investigator award. 
}\\
Georgia Tech\\
\texttt{vempala@gatech.edu}}

\begin{document}
\maketitle

\begin{abstract}
We prove nearly quadratic lower bounds on the complexity of randomized algorithms for
linear optimization and uniform sampling over convex bodies in the
membership oracle model. For linear optimization, this matches the known nearly
quadratic upper bound up to a polylog factor in the dimension. For uniform 
sampling, this improves on the previous linear lower bound.
Our construction also implies the same lower bound for volume estimation. 
\end{abstract}

\tableofcontents
\thispagestyle{empty}
\vfill

\newpage
\section{Introduction}

\setcounter{page}{1}
Optimizing a linear function over a convex body is a classical
algorithmic problem.
Following the development of the Ellipsoid method by Yudin and Nemirovski~\cite{YN76} and Shor~\cite{Shor77}, 
Khachiyan's breakthrough analysis gave a polynomial-time algorithm for linear programming~\cite{Kh80}. 
The foundational work of
Gr\"{o}tschel, Lov\'{a}sz, and Schrijver~\cite{GLS}
developed a systematic theory of oracles and reductions for the general setting of 
convex bodies.
A \emph{membership oracle} for a convex body $K$ answers queries of
the form ``$x\in K$?''. The algorithm is also given an interior point
$x_0$ and bounds $r,R>0$ such that
$x_0+rB^n\subseteq K\subseteq RB^n$.
In this framework, approximate linear optimization can be achieved with
query complexity and additional arithmetic complexity polynomial in the dimension 
and in the logarithms of $R/r$ and the inverse accuracy. More than three decades later,
following advances in cutting-plane methods~\cite{Vai,lee2015faster} and other improvements~\cite{BV04,KV06}, 
Lee, Sidford, and Vempala~\cite{LSV19}
gave a randomized algorithm with nearly quadratic
query complexity in the dimension.

A second fundamental problem in the oracle model is uniformly
sampling a convex body. The celebrated work of
Dyer, Frieze, and Kannan~\cite{DyerFK89} gave a randomized
polynomial-time sampling algorithm and used it to approximate the
volume of a convex body to arbitrary relative accuracy. The latter
consequence is particularly striking in the context of exponential
lower bounds for \emph{deterministic} volume
algorithms~\cite{E86,Barany1987}. A long and active line of work has
extended these methods to sampling and integration of logconcave
densities and substantially reduced their query
complexity~\cite{L90,ApplegateK91,LS93,KLS97,frieze1999log,LV07,LV06,CV18,JLLV26,kook2025sampling}. 
The previously known lower bounds for randomized linear
optimization and uniform
sampling are linear in the dimension,
raising a basic question:
\begin{center}
\emph{Are there superlinear lower bounds for randomized algorithms
for linear optimization and uniform sampling in the membership
oracle model?}
\end{center}
In this paper, we prove nearly quadratic lower bounds for
both problems, and obtain a lower bound for volume estimation as a corollary. 
The ``hard'' bodies are centrally symmetric
parallelepipeds with polynomially bounded ratio of outer to inner
radius. 

\subsection{Results}

For a sufficiently large universal constant $C$, consider the class
\[
 \mathcal C_n=\{K\subset\R^n:\ K\text{ is a convex body and }
             B^n\subseteq K\subseteq Cn^{3/2}B^n\}.
\]

In the optimization, sampling, and volume results below, the expectation is over the algorithm's randomness. We begin with optimization.
\begin{theorem}[Linear optimization]
\label{thm:intro-optimization}
For $K\in\mathcal C_n$ and unit vector $c$, set
$\OPT=\max_{x\in K}c^\top x$. Any randomized algorithm that, on any 
such input $(K,c)$, with probability at least $2/3$, outputs
$Z >0$ satisfying either
\[
\frac{1}{2} \OPT\le Z\le\OPT
 \qquad\text{or}\qquad
 \OPT-\,n\,\le Z\le\OPT
\]
requires $\Omega(n^2/\log^2 n)$ membership queries in expectation
on some input $(K,c)$. 
\end{theorem}
\paragraph{Consequences.}
Our lower bound matches the nearly quadratic upper bound of~\cite{LSV19}, up to a
polylog factor, even for constant-factor
approximation of the optimum value.

It also has consequences for the five weak
oracles of Gr\"otschel, Lov\'asz, and Schrijver~\cite{GLS}.
Recall that membership (MEM) decides whether a point belongs to the body; separation
(SEP) also gives a separating hyperplane for points outside the body.
Optimization (OPT) returns an approximate maximizer of a linear
functional. Validity (VAL) decides if a linear inequality holds
for the body; violation (VIOL) also gives a violating point
when it does not.

The diagram below summarizes the known query bounds, including the
$\mathrm{SEP}\to\mathrm{MEM}$ bound of Lee, Sidford, and
Vempala~\cite[Theorem~21]{LSV19}.
An arrow $A\to B$ denotes implementing oracle $A$ using oracle $B$\footnote{We use the
accuracy and confidence conventions of~\cite[Section~2.1]{LSV19}. The radius ratio and inverse accuracy
are polynomially bounded in $n$; reductions may request inverse-polynomial
accuracy and failure probability from their source oracles. The notation
$\widetilde O$ suppresses polylogarithmic factors in these parameters.}.
Solid and dashed arrows denote $\widetilde O(1)$ and $\widetilde O(n)$
queries, respectively.
\begin{equation*}
\begin{tikzpicture}[
 baseline=(sep.base), x=2.75cm,
 every node/.style={font=\large,inner xsep=5pt,inner ysep=3pt},
 short/.style={-{Stealth[length=2mm,width=1.7mm]},line width=1.1pt},
 linear/.style={-{Stealth[length=1.8mm,width=1.4mm]},
                line width=.55pt,dash pattern=on 3pt off 2pt}
]
 \node (mem) at (0,0) {$\mathrm{MEM}$};
 \node (sep) at (1,0) {$\mathrm{SEP}$};
 \node (opt) at (2,0) {$\mathrm{OPT}$};
 \node (viol) at (3,0) {$\mathrm{VIOL}$};
 \node (val) at (4,0) {$\mathrm{VAL}$};
 \draw[short] ([yshift=3pt]mem.east) -- ([yshift=3pt]sep.west);
 \draw[linear] ([yshift=-3pt]sep.west) -- ([yshift=-3pt]mem.east);
 \draw[linear] ([yshift=3pt]sep.east) -- ([yshift=3pt]opt.west);
 \draw[linear] ([yshift=-3pt]opt.west) -- ([yshift=-3pt]sep.east);
 \draw[short] ([yshift=3pt]opt.east) -- ([yshift=3pt]viol.west);
 \draw[short] ([yshift=-3pt]viol.west) -- ([yshift=-3pt]opt.east);
 \draw[linear] ([yshift=3pt]viol.east) -- ([yshift=3pt]val.west);
 \draw[short] ([yshift=-3pt]val.west) -- ([yshift=-3pt]viol.east);
\end{tikzpicture}
\end{equation*}
By composition and polarity (see \cite{GLS} or \cite[Section~6]{AGGW20}),
Theorem~\ref{thm:intro-optimization} shows that the resulting bounds for
reductions between every pair of oracles are optimal up to polylogarithmic factors.

Next, with quantum access to membership, linear optimization admits
algorithms using only $\widetilde O(n)$ queries~\cite{CCLW20,AGGW20}.
Together with Theorem~\ref{thm:intro-optimization}, this gives a 
separation between 
classical and quantum query complexities of optimization from membership.

Kerger~\cite{Kerger2026} presented an AI proof of a nearly quadratic lower bound
for deterministic convex minimization with an exact function value oracle and 
a feasible region known in advance. This nearly matches Protasov's upper bound~\cite{Protasov1996}. 
Our lower bound is for optimizing an explicit linear objective
over an unknown convex body given by a membership oracle and applies to randomized algorithms. 
The above consequences do not follow from \cite{Kerger2026}. 

\paragraph{Sampling.} Our construction gives the same nearly quadratic lower bound
for uniform sampling, improving the previous $\Omega(n)$ lower
bound~\cite{GoyalRademacherVempala2015}.
\begin{theorem}[Uniform sampling]
\label{thm:intro-sampling}
Any randomized algorithm whose output distribution is within total variation
distance $1/10$ of the uniform distribution on $K$, for every $K\in\mathcal C_n$,
requires $\Omega(n^2/\log^2 n)$ membership queries in expectation
on some input $K\in\mathcal C_n$. 
\end{theorem}
We remark that for this sampling lower bound we can restrict further to bodies satisfying
$B^n \subset K \subset O(n)B^n$. We keep the higher
polynomial for the sake of a unified proof.

Finally, the same construction gives a nearly quadratic lower bound
for volume estimation, recovering the result in~\cite{RV08}.
\begin{corollary}[Volume estimation]
\label{cor:intro-volume}
Any randomized algorithm that, for any $K\in\mathcal C_n$, with
probability at least $2/3$, outputs $\ell,u > 0$ such that
$\ell\le\operatorname{vol}(K)\le u$ and $u/\ell\le2$,
requires $\Omega(n^2/\log^2 n)$ membership queries in expectation
on some input $K\in\mathcal C_n$. 
\end{corollary}

All our lower bounds hold with constant probability, not just in expectation: 
for some input,
the query count $T$ satisfies
\[
 \Prob\left( T\ge a\,\frac{n^2}{\log^2 n}\right)\ge\frac1{100}
\]
for some universal constant $a>0$. The probability is over
the algorithm's randomness.

\subsection{The hard distribution}
\label{sec:hard-distribution}

Let $R$
be the matrix with rows $R_1^\top,\ldots,R_n^\top$, where each $R_i$ is drawn 
independently from $N(0,I_n/n)$. Let $L$ be a sufficiently large
universal constant times $n$, chosen in Section~\ref{sec:oracle-complexity}.
For $b_1=1$, $b_2=3$, and $j\in\{1,2\}$, define
\[
 K_j(R)=\{x\in\R^n:\ |\langle R_i,x\rangle|\le1\ (i<n),\
                    |\langle R_n,x\rangle|\le b_j L\}.
\]
The two bodies differ only in the last constraint, whose bound is $L$ or $3L$.

Almost surely $R$ is invertible. The vector $R^{-1}e_n$ is normal
to the first $n-1$ rows and parallel to the ``long'' direction of the parallelopiped. 
The last constraint
gives the two bodies different lengths in this direction (Fig.~\ref{fig:hard-distribution}).
The two oracles disagree only on queries in $K_2(R)\setminus K_1(R)$.

\begin{figure}[htbp]
\centering
\includegraphics[width=0.8\linewidth,trim=15bp 130bp 40bp 30bp,clip]{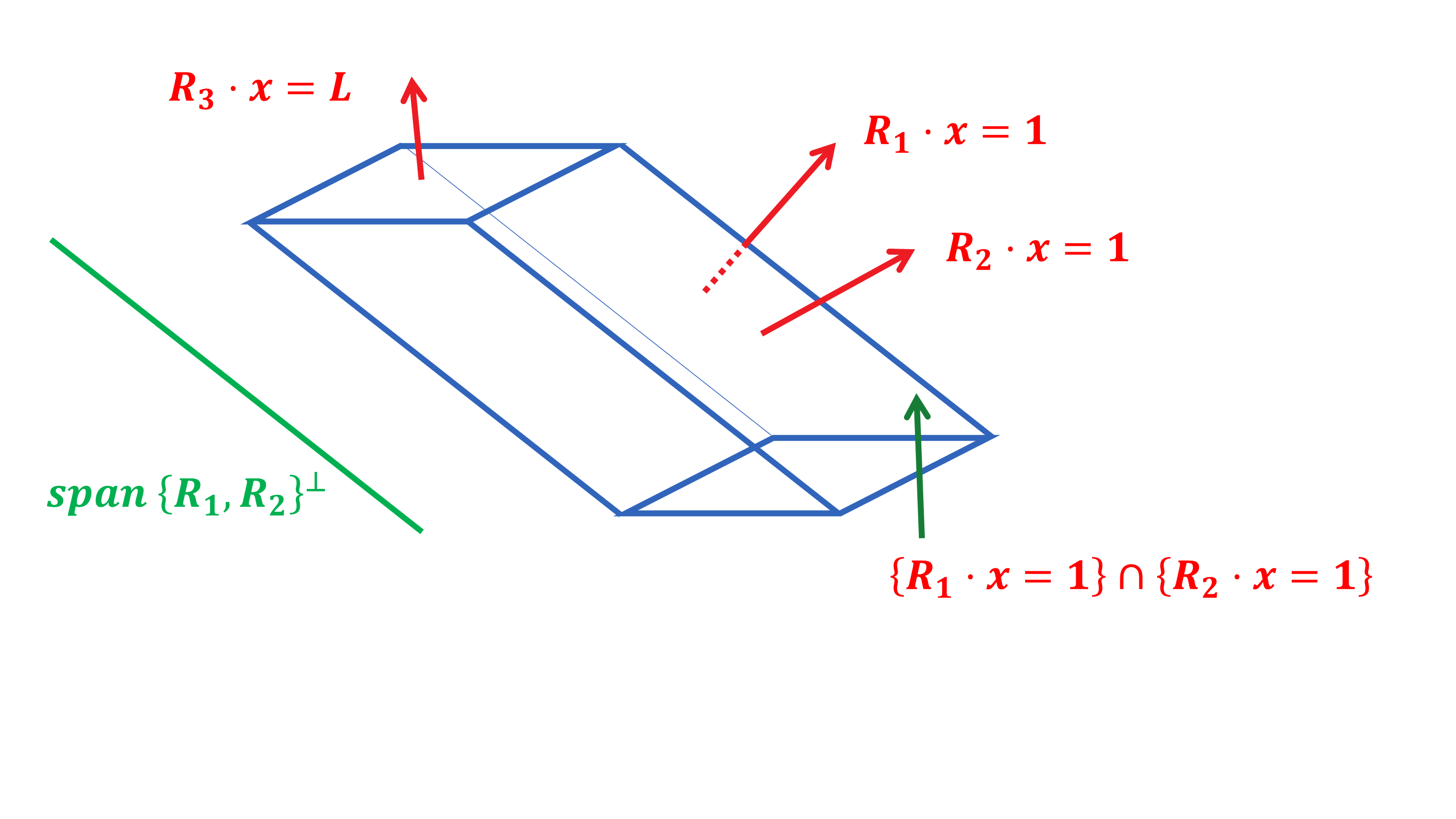}
\caption{The body $K_1(R)$ for $n=3$. Its long direction is
orthogonal to $R_1$ and $R_2$, while $R_3$ is normal to the top and bottom faces.}
\label{fig:hard-distribution}
\end{figure}

\begin{theorem}[Distinguishing the hard pair]
\label{thm:intro-hard-pair-distinguishing}
Let $j$ be uniform on $\{1,2\}$ and independent of $R$. 
Any randomized algorithm that, with probability at least $2/3$, 
identifies $j$
from a membership oracle for $K_j(R)$ makes 
$\Omega(n^2/\log^2 n)$ queries in expectation, where probability and expectation are
over $R,j$ and the algorithm's internal randomness. 
\end{theorem}

\subsection{Product partitions}
\label{sec:product-partitions}

Our starting point is Yao's principle~\cite{Yao1977}: a lower bound for
deterministic decision trees under the input distribution yields a lower
bound for randomized algorithms. We reveal the last row $R_n$ to the algorithm. 
For each fixed $j$, once the algorithm's random
choices and $R_n$ are fixed, its execution is a deterministic decision tree on
$(R_1,\ldots,R_{n-1})$. Each internal node specifies a query, each
outgoing edge records an oracle response, and each leaf records a
terminated execution. A \emph{transcript} is the sequence of queries
and responses on such a path.

For the purpose of analysis, we strengthen the membership oracle as follows: 
with a NO response, it also reports the
first violated row constraint and the sign of the violation. It checks
$R_1,\ldots,R_{n-1}$ in order before $R_n$. Each query now has
at most $2n+1$ possible responses. Giving $R_n$ and the additional oracle
responses can only help the algorithm, so a lower bound in this model
applies to the original one. The advantage is that every response
restricts the hidden rows {\em independently}. For a query $q$, a YES response
imposes the strips
\[
                  |\langle R_i,q\rangle|\le1,\qquad i<n.
\]
A NO response identifying a hidden row $i$ imposes these strips on the preceding
rows and a halfspace, determined by the reported sign, on $R_i$;
it places no new restriction on later hidden rows. A NO response
identifying the public (last) row certifies that every hidden constraint is satisfied.
Intersecting these restrictions
along a path gives a set
\[
                         P=F_1\times\cdots\times F_{n-1},
\]
where each $F_i\subseteq\R^n$ is convex; in fact it is an intersection of halfspaces. 
Distinct leaves give disjoint product
sets. We call each such nonempty $P$ a \emph{part} and
each constituent $F_i$ a \emph{factor}: $F_i$ is the set of values
of row $R_i$ consistent with the transcript leading to $P$. For a part with 
positive measure, the rows remain independent, and the distribution of
any $R_i$ is its original Gaussian {\em restricted} to $F_i$.

We call a query \emph{separating} if it lies in $K_2(R)\setminus K_1(R)$. On an event of high Gaussian
probability, every separating query points close, up to sign, to the hidden
normal. Stopping at the first separating query fixes
a candidate direction. We will show that, 
after discarding parts of small total measure, the normal is
localized around this direction with high probability.
Our main structural lemma states that any substantial family of
normal-localizing product parts contains a substantial subfamily of
exponentially small parts.

Let $\gamma_n=N(0,I_n)$ be the standard $n$-dimensional Gaussian and 
$\bar\gamma_n=N(0,I_n/n)$.
For a nonzero vector $v$, write $\overrightarrow{v}=v/\|v\|$. For nonzero vectors 
$v,w$, define the
{\em projective} distance between them as
\[
 \distp(v,w)=
 \min\bigl\{\|\overrightarrow{v}-\overrightarrow{w}\|,
             \|\overrightarrow{v}+\overrightarrow{w}\|\bigr\}.
\]
For a full-row-rank matrix $A$ with rows
$R_1^\top,\ldots,R_{n-1}^\top$, let $v(A)$ be a unit vector normal
to the row span of $A$ (its sign will not matter). 
The \emph{measure} of a part $P$ is the probability
$\bar\gamma_n^{\otimes(n-1)}(P)$. We can now state the main structural 
lemma. 

\begin{lemma}[Normal-localizing product partition]
\label{thm:intro-main-partition}
There are universal constants $a>0$ and $0<c<1$ such that the following
holds for every $\eta\in(0,1)$.
Let $\mathcal P$ be a countable disjoint
partition of $(\mathbb R^n)^{n-1}$ into convex product parts
\[
P=F_1\times\cdots\times F_{n-1}.
\]
Let $\mathcal S_{\rm loc}\subset\mathcal P$ be a family of
parts with total
$\bar\gamma_n^{\otimes(n-1)}$-measure at least $\eta$.
Assume that each positive-measure $P\in\mathcal S_{\rm loc}$ has a fixed
unit vector $u(P)$ such that, when $A$ is sampled from the 
restriction of $\bar\gamma_n^{\otimes(n-1)}$ to $P$,
\[
\Prob\left(
\distp(v(A),u(P))
\le
c\eta
\right)
\ge
1-\frac{1}{1024}.
\]
Then there is a subfamily
$\mathcal S'\subset\mathcal S_{\rm loc}$ of total measure at least
$\eta/2$ such that every $P\in\mathcal S'$ satisfies
\[
\bar\gamma_n^{\otimes(n-1)}(P)
\le
\exp\left(-a\frac{n^2}{\log n}\right).
\]
\end{lemma}
Since the parts of the product partition are disjoint, the lemma
implies
\[
|\mathcal S_{\rm loc}|
\ge
\frac{\eta}{2}\exp\left(a\frac{n^2}{\log n}\right).
\]

\begin{corollary}
\label{cor:variable-depth-partition}
Under the hypotheses of Lemma~\ref{thm:intro-main-partition},
suppose the parts are the leaf sets of a deterministic decision tree
with at most $2n+1$ possible responses to each query.
Then there is a universal constant $c>0$ and, for each $\eta$, a threshold $n_\eta>0$ such that, for $n\ge n_\eta$,
the number of queries $T$ satisfies
\[
 \Prob\left(T\ge c\frac{n^2}{\log^2 n}\right)\ge\frac{\eta}{4}. 
\]
\end{corollary}
This immediately gives $\E T=\Omega(\eta n^2/\log^2 n)$.
Such product partitions were studied in earlier work~\cite{RV08,RademacherShao2008}.
Our approach follows the sampling-to-normal-localization strategy proposed by
Goyal, Rademacher and Vempala~\cite[Appendix]{GoyalRademacherVempala2015}.

\subsection{Proof overview}

The proof of the product-partition lemma has three main steps: 
normal localization
implies sign stability, sign stability forces large relative entropy of the row 
distributions, and
large relative entropy implies small measure.
Small part measure implies many parts and a query lower bound for finding a separating query. 
The oracle reductions follow by showing that successful algorithms 
must find separating
points with constant probability.

\paragraph{1. Normal localization implies sign stability.}
Fix a normal-localizing part $P$ and its unit vector $u=u(P)$.
Conditional on $P$, the rows remain independent. Identify $u^\perp$ with
$\R^{n-1}$ using fixed orthonormal coordinates, and set
\[
 X_i=P_{u^\perp}R_i\in\R^{n-1}.
\]
Let $X$ be the matrix
with rows $X_i^\top$. Now consider replacing a uniformly chosen row by an
independent sample from its distribution on the part.
Let $\nu$ be the probability that the sign of $\det X$ changes under this replacement.
We call the sign \emph{stable} if $\nu<1/8$.
For {\em unrestricted} Gaussian rows, the original and replacement rows
independently fall on either side of the span of other rows with equal
probability, so $\nu=1/2$.

Now, to understand the effect of normal localization, 
let us condition on all rows except $R_i$.
Project $R_i$ onto the plane perpendicular to the span of all other rows.
Normal localization forces this residual to be 
close in direction to the line in that
plane perpendicular to $u$; the two directions along the line correspond
to the two determinant signs. If the sign is unstable, then for a constant
fraction of choices of the index and the other rows, both signs retain
significant mass under localization. The covariance bound on the row distribution 
places mass on both sides
inside a narrow strip, and logconcavity within that strip forces mass near
the origin. Lemma~\ref{lem:localization-sign} shows that, on an unstable part with
sufficient normal localization, a uniformly chosen row has residual
length at most $t/\sqrt n$ with probability $\Omega(t)$.

For unrestricted Gaussian rows, the same short-residual event has
probability at most $t^2/2$. Crucially, this event does not depend on $u(P)$.
Averaging over the partition therefore bounds the total measure of
unstable localizing parts by $O(t)$. Choosing $t$ small compared with the
localizing family's measure leaves at least half that measure on
sign-stable parts.

\paragraph{2. Sign stability implies large relative entropy.}
The Gaussian reference for the projected rows is $N(0,I_{n-1}/n)$.
Scale the rows and this reference by $\sqrt n$, making the reference
the standard Gaussian while preserving relative entropy and flip probabilities.
For arbitrary independent rows with distributions $\mu_i$ of finite Gaussian
relative entropy $\mathcal K(\mu_i)$, we prove that
\[
 K = \sum_{i=1}^{n-1} \mathcal K(\mu_i)\le c\frac{(n-1)^2}{\log (n-1)}\quad\Longrightarrow\quad\nu\ge\frac18.
\]

We use induction on the dimension, going from $d=n-1$ to
$d_0=\lceil100\log(n-1)\rceil$ by repeatedly deleting one row and projecting down one dimension. 
We will see that deleting a uniformly chosen row and projecting
the others onto its orthogonal complement preserves the expected flip
probability. Using $K_d$ for the total relative entropy in dimension $d$, 
we observe that deletion removes $K_d/d$ entropy in expectation. 

We seek an additional decrease of nearly $K_d/d$ 
in relative entropy from projection. We first apply a
common invertible linear map to all rows
and then independently replace each row's length by that of a standard Gaussian.
We show that the map can be chosen so that
the average second-moment matrix is close to the identity in operator norm,
without increasing total entropy. 

If every row has entropy at most $d/d_0$, we delete a uniformly chosen row.
Using the second-moment bound, we can then apply the Gaussian Brascamp--Lieb inequality
(Lemma~\ref{lem:gaussian-entropic-bl}) to show the decrease in entropy, 
and $K_d/d^2$ increases by at most a factor $1+20/(d_0d)$ in
expectation. If some row has entropy greater than $d/d_0$, we delete
such a row. When $K_d<d^2/(4d_0)$, the resulting entropy decrease 
compensates for any increase in the expected flip probability.

We combine these considerations via the potential
\[
 \psi_d=2\nu+4d_0\frac{K_d}{d^2},
\]
observe that $\psi_{d_0}\ge1$, and induction gives $\psi_{n-1}>1/2$.
Consequently, $K\le(n-1)^2/(16d_0)$ implies $\nu\ge1/8$, proving the
claimed entropy bound for a sufficiently small universal $c$.

\paragraph{3. Large relative entropy implies small part measure.}
A Gaussian restricted to $F_i$ has relative entropy
$\log(1/\bar\gamma_n(F_i))$ with respect to the original Gaussian.
Applying the data processing inequality
(Lemma~\ref{lem:entropy-facts}~(\ref{item:data-processing}))
to the scaled projection therefore gives
\[
 K\le\sum_i\log\frac1{\bar\gamma_n(F_i)}=\log\frac1p,
 \qquad p=\bar\gamma_n^{\otimes(n-1)}(P) = \prod_{i=1}^{n-1}\bar\gamma_n(F_i).
\]
Sign stability forces $K\ge c n^2/\log n$, so every retained part satisfies
\[
 p\le\exp(-c n^2/\log n).
\]
Together, these parts contain at least half the original normal-localizing
family's measure.

\paragraph{Acknowledgements.} We are indebted to Luis Rademacher, Navin Goyal and Yin Tat Lee
for many helpful discussions. 
\section{Preliminaries}
\label{sec:preliminaries}

Throughout, the ambient dimension $n$ is assumed to exceed a sufficiently
large absolute constant. Fixed vectors use lowercase letters;
random variables, including the rows $R_i$ and their projections $X_i$,
use uppercase letters, as do matrices and subspaces. We write $v(A)$
for the normal determined by the row matrix $A$.
We use $\|\cdot\|$ for the Euclidean norm and
$\|A\|_{\rm op}=\sup_{\|v\|=1}\|Av\|$ for the operator norm.
We write $\mu(X)$
for the distribution of $X$ and $T_\#\mu$ for the distribution of $T(X)$ when
$X\sim\mu$. For a Euclidean subspace $L$, $P_L$ is orthogonal
projection to $L$ and $\gamma_L$ is the standard Gaussian distribution on $L$.
In $\R^d$, this is $\gamma_d=N(0,I_d)$, with density
$(2\pi)^{-d/2}e^{-\|x\|^2/2}$. We also use
$\bar\gamma_n=N(0,I_n/n)$.
For symmetric matrices, $A\preceq B$ means
$v^\top Av\le v^\top Bv$ for every $v$.

\paragraph{Relative entropy.}
All logarithms are natural. For probability measures $\mu,\nu$, their relative entropy is
\[
 D_{\rm KL}(\mu\Vert\nu)=\int \frac{d\mu}{d\nu}\log \frac{d\mu}{d\nu}\,d\nu = \int \log \frac{d\mu}{d\nu}\,d\mu.
\]
If $\mu$ is not absolutely continuous with respect to $\nu$, we set
$D_{\rm KL}(\mu\Vert\nu)=+\infty$.
For a probability measure $\mu$ on $\R^d$, its Gaussian relative entropy is
\[
                  \mathcal K(\mu)=D_{\rm KL}(\mu\Vert\gamma_d).
\]
The total variation distance is
$\|\mu-\nu\|_{\rm TV}=\sup_E|\mu(E)-\nu(E)|$, the supremum being
over measurable sets.

\begin{lemma}
\label{lem:entropy-facts}
Let $\mu,\nu$ be probability measures.
\begin{enumerate}
\renewcommand{\theenumi}{\alph{enumi}}
\renewcommand{\labelenumi}{(\theenumi)}
\item\label{item:entropy-variational}
Relative entropy is nonnegative and has the variational representation
\[
 D_{\rm KL}(\mu\Vert\nu)
   =\sup_g\left\{\E_\mu g-\log\E_\nu e^g\right\},
\]
where the supremum is over bounded measurable functions.
\item\label{item:data-processing}
Data processing inequality: For any measurable map $T$,
\[
 D_{\rm KL}(T_\#\mu\Vert T_\#\nu)
     \le D_{\rm KL}(\mu\Vert\nu).
\]
\item\label{item:product-entropy}
For finite products,
\[
 D_{\rm KL}\left(\bigotimes_i\mu_i\Big\Vert\bigotimes_i\nu_i\right)
     =\sum_iD_{\rm KL}(\mu_i\Vert\nu_i).
\]
\item\label{item:restriction-entropy}
For any measurable set $F$ with $\nu(F)>0$,
\[
 D_{\rm KL}(\nu(\cdot\mid F)\Vert\nu)
     =\log\frac1{\nu(F)}.
\]
\item\label{item:pinsker}
Pinsker's inequality: 
\[
 \|\mu-\nu\|_{\rm TV}
     \le\sqrt{D_{\rm KL}(\mu\Vert\nu)/2}.
\]
\end{enumerate}
\end{lemma}

See~\cite[Chapter~2]{PolyanskiyWu2025} for the basic
entropy properties, and~\cite[Theorems~4.6 and~7.10]{PolyanskiyWu2025}
for the variational formula and Pinsker's inequality.

\paragraph{Moments and linear transformation.}
Let $X\in\R^d$ have finite Gaussian relative entropy. Applying
Lemma~\ref{lem:entropy-facts}~(\ref{item:entropy-variational}) to bounded
truncations of $t\|x\|^2$ gives, for $0<t<1/2$,
\begin{equation*}
 t\E\|X\|^2\le\mathcal K(\mu(X))
                       -\frac d2\log(1-2t).
\end{equation*}
In particular, $X$ has finite second moments. Let $f$ be the density of
$X$. For an invertible deterministic linear map $A$, the density of $AX$ is
\[
 f_A(y)=\frac{f(A^{-1}y)}{|\det A|}.
\]
Substituting this density and the standard Gaussian density into the
definition of relative entropy gives
\begin{equation}\label{eq:linear-entropy-change}
 \mathcal K(\mu(AX))
 =\mathcal K(\mu(X))
   +\frac12\E\bigl[\|AX\|^2-\|X\|^2\bigr]-\log|\det A|.
\end{equation}

\begin{lemma}[Gaussian Brascamp--Lieb]
\label{lem:gaussian-entropic-bl}
Let $L_1,\ldots,L_N$ be nonzero subspaces of $\R^d$ and let
$c_1,\ldots,c_N\ge0$ satisfy
$\sum_s c_sP_{L_s}\preceq I_d$. For every probability measure
$\mu$ with finite $\mathcal K(\mu)$,
\[
 \sum_s c_sD_{\rm KL}((P_{L_s})_\#\mu\Vert\gamma_{L_s})
                    \le\mathcal K(\mu).
\]
\end{lemma}

See Carlen--Lieb~\cite[Theorem~4.2]{CarlenLieb2008} for this entropy form
of Gaussian Brascamp--Lieb, and
Carlen--Cordero-Erausquin~\cite{CarlenCorderoErausquin2009} for the
general relation between entropy subadditivity and Brascamp--Lieb
inequalities.

\paragraph{Logconcavity.}
A nonnegative function $f:\R^n \rightarrow \R_{\ge 0}$ is logconcave if
\[
 f((1-t)x+ty)\ge f(x)^{1-t}f(y)^t\qquad(0<t<1).
\]
A probability distribution is logconcave if it has a logconcave density. 

\begin{theorem}[Dinghas--Leindler--Pr\'ekopa {\cite{Din57,Lei72,Prekopa1973,Pr73b}}]
\label{thm:prekopa-marginal}
All marginals of an integrable logconcave function are logconcave.
\end{theorem}

\begin{lemma}
\label{lem:convex-gaussian-covariance}
If $X\sim N(0,I_n/n)$ restricted to a convex set of positive
Gaussian measure, then $X$ is full-dimensional and logconcave,
and $\operatorname{Cov}(X)\preceq I_n/n$.
\end{lemma}

The covariance bound follows from the classical Brascamp--Lieb
inequality~\cite{BL76}.

\paragraph{Gaussian estimates.}
\label{sec:gaussian-estimates}
\begin{lemma}
\label{lem:gaussian-estimates}
Let $Z\sim N(0,1)$, $G\sim\gamma_d$, and let $R$ be an
$n$-by-$n$ matrix with independent $N(0,1/n)$ entries.
\begin{enumerate}
\renewcommand{\theenumi}{\alph{enumi}}
\renewcommand{\labelenumi}{(\theenumi)}
\item\label{item:gaussian-scalar-tail}
For $t\ge0$,
\[
 \Prob(|Z|\ge t)\le2e^{-t^2/2}.
\]
\item\label{item:gaussian-norm-tail}
For $h>0$,
\[
 \Prob(\|G\|^2>d(1+h))
   \le e^{-d(h-\log(1+h))/2}.
\]
\item\label{item:gaussian-projection}
For any fixed $r$-dimensional subspace $L\subseteq\R^d$,
$\|P_LG\|^2$ has the $\chi_r^2$ distribution. The same holds
conditionally on $L$ if $L$ is random and independent of $G$.
In particular,
\[
                    \Prob(\chi_2^2\le u)=1-e^{-u/2}\quad(u\ge0).
\]
\item\label{item:gaussian-inverse-tail}
For $x>0$,
\[
 \Prob(\|R^{-1}\|_{\rm op}\ge x)\le\frac{n}{x}.
\]
\end{enumerate}
\end{lemma}

\begin{proof}
For (\ref{item:gaussian-scalar-tail}) and
(\ref{item:gaussian-norm-tail}), see
\cite[Sections~2.2, 2.4, and~2.5]{BoucheronLugosiMassart2013};
for (\ref{item:gaussian-projection}), see
\cite[Sections~2.4 and~3.3]{Anderson2003}.

For an $n$-by-$n$ matrix $A$ with independent $N(0,1)$ entries,
Edelman's formula for the least eigenvalue of $AA^\top$
\cite[Theorem~3.1]{Edelman1988} implies
\[
 \Prob(\|A^{-1}\|_{\rm op}\ge t)\le\frac{\sqrt n}{t}
 \qquad(t>0).
\]
Taking $A=\sqrt n R$ and $t=x/\sqrt n$ proves
part (\ref{item:gaussian-inverse-tail}).
\end{proof}

\section{The product partition lower bound}
\label{sec:product-partition-bound}
Let $X$ be a $d\times d$ random matrix whose rows $X_i^\top$ are independent and
have full-dimensional densities. Write $f(X)=\operatorname{sign}\det X$,
and let $X^{(i)}$ be obtained by replacing row $i$ by an independent copy.
Define the individual and average \emph{flip probabilities} by
\[
 \nu_i=\Prob(f(X)\ne f(X^{(i)})),\qquad
 \nu=\frac1d\sum_{i=1}^d\nu_i.
\]
Thus $\nu$ is the probability that replacing a uniformly chosen row changes
the determinant sign. Conditional on other rows, the original and
replacement signs are independent, each positive with some probability $p$.
They disagree with probability $2p(1-p)\le1/2$, so
$0\le\nu_i,\nu\le1/2$.

\subsection{Normal localization implies sign stability}
\label{sec:normal-to-sign}
Let $A$ have rows $R_1^\top,\ldots,R_{n-1}^\top \in\mathbb R^n$. 
Fix a unit vector $u\in\mathbb R^n$. 
Projecting the rows of $A$ to $u^\perp$ gives a square matrix 
in fixed orthonormal coordinates. Denote its
determinant sign by $f_u(A)$ and its average flip probability by $\nu_u$.
We will bound the total measure of parts on which the normal is localized
but the projected determinant sign is {\em unstable}, i.e., the flip probability 
is at least $1/8$. The next lemma shows that,
on such a part, a row is likely to lie close to the span of the others.

\begin{lemma}
\label{lem:localization-sign}
There are universal constants $0<c<1$ and $C\ge1$ with the following property
for every $0<t<1$.
Let $R_1,\ldots,R_{n-1}\in\R^n$
be independent full-dimensional logconcave rows with
$\operatorname{Cov}(R_i)\preceq I_n/n$.
Suppose, for some fixed unit vector $u$,
\[
 \Prob\left(\distp(v(A),u)\le\frac{t}{C}\right)
       \ge1-\frac1{1024}.
\]
If the projected determinant has flip probability $\nu_u$ at least $1/8$, 
then
\begin{equation*}
 \frac1{n-1}\sum_{i=1}^{n-1}
 \Prob\left(\dist\bigl(R_i,\operatorname{span}\{R_j:j\ne i\}\bigr)
                     \le\frac t{\sqrt n}\right)
       \ge ct.
\end{equation*}
\end{lemma}

\begin{proof}
Let $\delta=1/1024$ and
\[
 \varepsilon=\frac{t}{C},\qquad
 \mathcal E=\{\distp(v(A),u)\le\varepsilon\}.
\]
We will choose $C\ge2$ later. Thus $\varepsilon\le1/2$ and
$\Prob(\mathcal E^c)\le\delta$.

Choose $i$ uniformly from $\{1,\ldots,n-1\}$, independently of the rows,
and condition on $(i,R_{-i})$, where $R_{-i}=(R_j)_{j\ne i}$. Almost surely,
$H=\operatorname{span}(R_{-i})$ has dimension $n-2$ and $u\notin H$.
In the plane $L=H^\perp$, choose an orthonormal basis $(e,z)$ with
$z\perp u$, and write the residual of row $i$ as
\[
 q=P_LR_i=Te+Zz.
\]
Put $P=P_{u^\perp}$. The vectors $PR_j$, $j\ne i$, span a hyperplane
in $u^\perp$ perpendicular to $z$, and $\langle PR_i,z\rangle=Z$.
By linearity in row $i$,
\[
\det(PR_1,\ldots,PR_i,\ldots,PR_{n-1})
=
Z\,\det(PR_1,\ldots,z,\ldots,PR_{n-1}).
\]
The second determinant is fixed and nonzero under the conditioning,
so $f_u(A)$ equals the sign of $Z$ up to a fixed sign.
Moreover, the unit normal $v=v(A)$ lies in $L$ and is perpendicular to $q$.
Hence, 
\[
 \frac{|T|}{\sqrt{T^2+Z^2}}=|\langle v,z\rangle|.
\]
Since $z\perp u$, the RHS is at most $\varepsilon$ on $\mathcal E$.
Therefore,
\begin{equation}\label{eq:localization-transverse-band}
 |T|\le\frac{\varepsilon}{\sqrt{1-\varepsilon^2}}|Z|
       \le3\varepsilon|Z|\qquad\text{on }\mathcal E.
\end{equation}
Let $\Prob_i$, $\E_i$, $\Var_i$ denote probability, expectation,
and variance conditional on $(i,R_{-i})$. Call a pair $(i,R_{-i})$ \emph{good} if
both determinant signs retain mass under localization:
\[
 \mathcal G=\left\{(i,R_{-i}):
   \Prob_i(\mathcal E,Z<0)\ge\delta,\quad
   \Prob_i(\mathcal E,Z>0)\ge\delta\right\}.
\]
Outside $\mathcal G$, one sign has probability at most
$\delta+\Prob_i(\mathcal E^c)$, so refreshing row $i$ flips the determinant sign with
probability at most $2\delta+2\Prob_i(\mathcal E^c)$.
On $\mathcal G$, the flip probability is at most $1/2$. Consequently,
\[
 \nu_u\le\frac12\Prob(\mathcal G)+2\delta+2\Prob(\mathcal E^c)
       \le\frac12\Prob(\mathcal G)+4\delta.
\]
Thus
\begin{equation}\label{eq:localization-useful-configurations}
 \Prob(\mathcal G)\ge2\nu_u-8\delta\ge\frac18.
\end{equation}

Fix a good pair. Independence leaves the distribution of $R_i$
unchanged, so $(T,Z)$ has a logconcave density by
Theorem~\ref{thm:prekopa-marginal}, and $\Var_i(Z)\le1/n$.
Both signs of $Z$ have probability at least $\delta$. Writing $m=\E_iZ$,
the half-line opposite $m$ gives
$\delta m^2\le\Var_i(Z)\le1/n$.
Hence $|m|\le1/\sqrt{\delta n}$, and Chebyshev's inequality gives
\[
 \Prob_i\left(|Z|>\frac3{\sqrt{\delta n}}\right)
 \le\Prob_i\left(|Z-m|>\frac2{\sqrt{\delta n}}\right)
 \le\frac\delta4.
\]
Each of $\mathcal E\cap\{-3/\sqrt{\delta n}\le Z\le0\}$ and
$\mathcal E\cap\{0\le Z\le3/\sqrt{\delta n}\}$ therefore has probability at least
$3\delta/4$. Choose $C\ge18/\sqrt\delta$. On these events, the cone bound
\eqref{eq:localization-transverse-band} gives
\[
 |T|\le\frac{9\varepsilon}{\sqrt{\delta n}}\le\frac{t}{2\sqrt n}.
\]
Hence the strip $S=\{|T|\le t/(2\sqrt n)\}$ contains at least $3\delta/4$ mass
in each large $Z$-half-interval (see Figure~\ref{fig:stability}):
\[
 \Prob_i\left(S,-\frac3{\sqrt{\delta n}}\le Z\le0\right),\quad
 \Prob_i\left(S,0\le Z\le\frac3{\sqrt{\delta n}}\right)\ge\frac{3\delta}{4}.
\]

\begin{figure}[htbp]
\centering
\includegraphics[width=\linewidth,trim=20bp 105bp 15bp 45bp,clip]{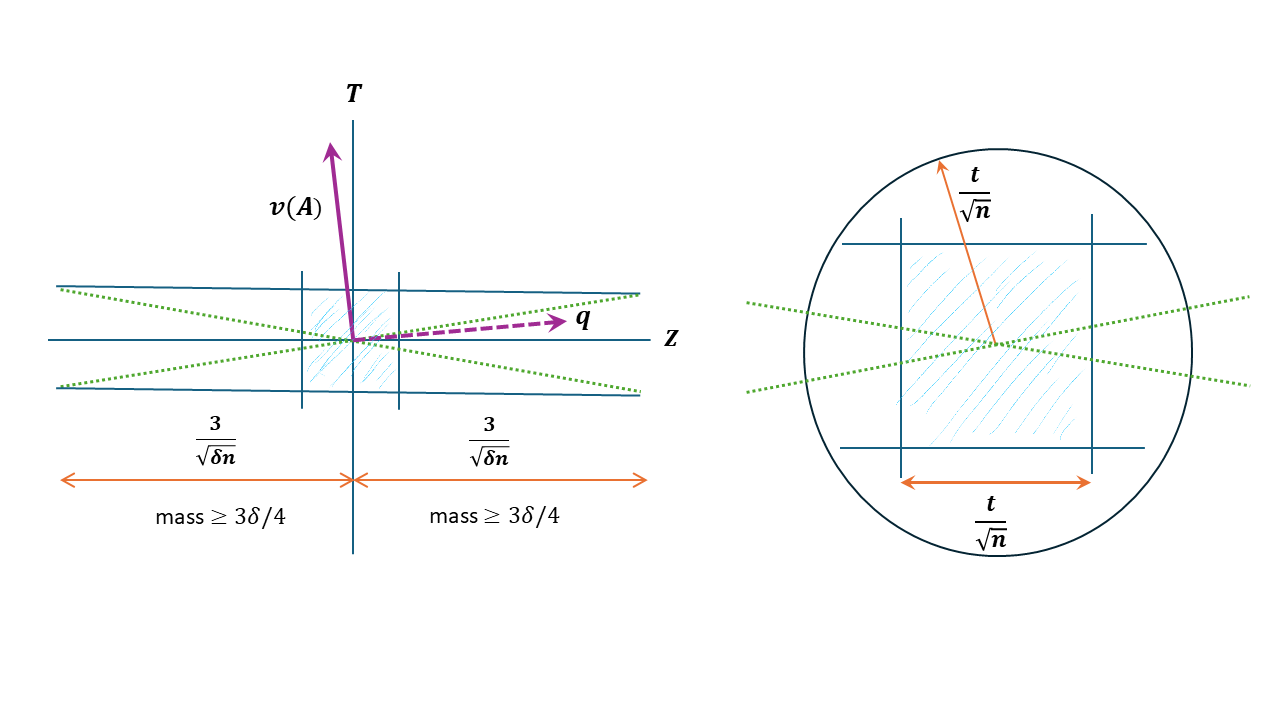}
\caption{For a good pair $(i,R_{-i})$, each large $Z$-half-interval of width
$3/\sqrt{\delta n}$ contains at least $3\delta/4$ mass within the
horizontal strip $|T|\le t/(2\sqrt n)$.
In this strip, logconcavity gives a lower bound proportional to $t$
for the mass in the central square of side $t/\sqrt n$.
The square lies in the disk of radius $t/\sqrt n$; for the
unrestricted Gaussian, this disk has mass at most $t^2/2$.}
\label{fig:stability}
\end{figure}

Conditioning on the convex strip $S$ and taking the $Z$-marginal preserves
logconcavity by Theorem~\ref{thm:prekopa-marginal}. Its density is therefore
nondecreasing toward zero on at least one half-line. On that side, the
average density on a shorter interval adjacent to zero is at least its
average on a longer such interval. Thus shortening the interval from
length $3/\sqrt{\delta n}$ to length $t/(2\sqrt n)$ retains at least
a $t\sqrt\delta/6$ fraction of its mass.
Multiplying by $\Prob_i(S)$ gives
\[
 \Prob_i\left(|T|\le\frac{t}{2\sqrt n},\ |Z|\le\frac{t}{2\sqrt n}\right)
 \ge\frac{t\sqrt\delta}{6}\,\frac{3\delta}{4}
 =\frac{\delta^{3/2}}8\,t.
\]
This square lies inside the disk $\{\|q\|\le t/\sqrt n\}$.
Averaging over the good pairs and using
\eqref{eq:localization-useful-configurations}, we obtain
\[
 \Prob\left(\|q\|\le\frac t{\sqrt n}\right)
 \ge\Prob(\mathcal G)\frac{\delta^{3/2}}8\,t
 \ge\frac{\delta^{3/2}}{64}\,t.
\]
This proves the lemma with $c=\delta^{3/2}/64$.
\end{proof}

For unrestricted Gaussian rows, the short-residual probability is at most
$t^2/2$. Comparing this with the lower bound $ct$ gives the following
consequence for a partition.

\begin{lemma}
\label{lem:stable-parts}
Let $c,C$ be the constants in Lemma~\ref{lem:localization-sign}, and
fix $0<t<1$. Let $\mathcal S$ be a family of positive-measure parts in
a countable convex product partition of $(\R^n)^{n-1}$.
Under the original Gaussian product distribution, suppose each
$P\in\mathcal S$ has a fixed unit vector $u(P)$ such that
\[
 \Prob\left(\distp(v(A),u(P))\le\frac{t}{C}
                    \;\middle|\; A\in P\right)
       \ge1-\frac1{1024}.
\]
One can discard parts of
$\mathcal S$ of total Gaussian measure at most $t/(2c)$
so that on any remaining part,
the rows projected onto $u(P)^\perp$ have determinant-sign flip
probability less than $1/8$.
\end{lemma}

\begin{proof}
Choose $i$ uniformly from $\{1,\ldots,n-1\}$, independently of the rows,
and let
\[
 E=\left\{
 \dist\bigl(R_i,\operatorname{span}\{R_j:j\ne i\}\bigr)
 \le\frac{t}{\sqrt n}
 \right\}.
\]
For unrestricted Gaussian rows, the orthogonal complement of the other
$n-2$ rows is almost surely two-dimensional. Conditional on these rows,
the residual has distribution $N(0,I_2/n)$ in orthonormal coordinates. Thus Lemma~\ref{lem:gaussian-estimates}~(\ref{item:gaussian-projection}) gives
\begin{equation*}
 \Prob(E)=1-e^{-\frac{t^2}{2}}\le\frac{t^2}{2}.
\end{equation*}

Let $B$ be the union of parts in $\mathcal S$ whose conditional flip
probability is at least $1/8$. On each such part $P$, Lemma~\ref{lem:convex-gaussian-covariance}
shows that the rows are full-dimensional and logconcave, with covariance
bounded by $I_n/n$. Lemma~\ref{lem:localization-sign} therefore gives
\[
 \Prob(E\mid A\in P)\ge ct.
\]
The event $E$ does not depend on $u(P)$, so averaging over these parts gives
\[
 ct\,\Prob(A\in B)
 \le\Prob(E\cap\{A\in B\})
 \le\Prob(E)
 \le\frac{t^2}{2}.
\]
Discarding these parts removes Gaussian measure at most $t/(2c)$,
and every remaining part has conditional flip probability less than $1/8$.
\end{proof}

\subsection{Sign stability implies large relative entropy}
\label{sec:entropy-induction}

Small Gaussian relative entropy forces a substantial flip probability.
The argument in this subsection applies to arbitrary independent rows
with finite Gaussian relative entropy.

\begin{theorem}
\label{thm:sign-entropy}
There is a universal constant $c>0$ such that, for every sufficiently
large $m$, the following holds. Let $X_1,\ldots,X_m\in\R^m$ be independent
random vectors with finite Gaussian relative entropy. Then
\[
 \sum_{i=1}^m\mathcal K(\mu(X_i))
       \le c\frac{m^2}{\log m}\quad\Longrightarrow\quad\nu\ge\frac18.
\]
\end{theorem}

Finite Gaussian relative entropy 
and independence ensure that the square row matrices below are nonsingular.
We first analyze how projection affects flip probabilities, then prepare
the normalization and entropy estimates for the main induction.

\begin{lemma}[Flip probability preservation]
\label{lem:projection-refresh}
For $d\ge2$, let $X_1,\ldots,X_d\in\R^d$ be independent with
full-dimensional densities. Fix an
index $r$, expose $X_r$, and project the other rows onto
$X_r^\perp$. Conditional on $X_r$, let $\nu'$ be the probability
that the determinant sign of the projected $(d-1)\times(d-1)$ matrix
changes when a uniformly chosen row $j\ne r$ is replaced by the
projection of an independent copy of $X_j$. The projection and
orthonormal coordinates on $X_r^\perp$ are held fixed during this
replacement. Then
\begin{equation}\label{eq:projection-refresh}
 \E_{X_r}\nu'=\frac1{d-1}\sum_{j\ne r}\nu_j
                 =\frac{d\nu-\nu_r}{d-1}.
\end{equation}
If $r$ is chosen uniformly and independently, then $\E\nu'=\nu$.
\end{lemma}

\begin{proof}
Fix a nonzero value $X_r=x$ and let $v=x/\|x\|$ and $e_d$ be the
last standard basis vector. Let $Q$ be an orthogonal transformation that makes row $r$ of $XQ$ equal to $(0,\ldots,\|x\|)$.  For $v \neq e_d$, let $w= (e_d-v)/\|e_d-v\|$. One such choice is
\[
Q = \begin{cases}
I - 2ww^\top \quad &\mbox{ if } v \neq e_d\\
I \quad &\mbox{ if }  v = e_d.
\end{cases} 
\]
For $v\neq e_d$, $Q$ is the reflection across the
hyperplane $(e_d-v)^\perp$.
We then have  $Q^\top Q = I$, $Qe_d=v$ and $x^TQ = \|x\|e_d^\top$.
Now let $Y$ be $XQ$ with row $r$ and the last
column deleted. Since the first $d-1$ columns of $Q$ form an orthonormal basis of $x^\perp$, the rows of $Y$ are the coordinates of the projected rows in this basis. Expanding along row $r$ gives
\[
 \det X=\frac{(-1)^{r+d}\|x\|}{\det Q} \det Y.
\]
Replacing any row $j\ne r$ leaves the factor on the RHS unchanged. Moreover,
conditioning on $X_r=x$ leaves the other rows and their independent
replacements with their original distributions. Thus the flip probability
of projected row $j$ is
\[
 \nu'_j(x)=\Prob(f(X)\ne f(X^{(j)})\mid X_r=x).
\]
For fixed $r$, averaging over $x$ and then uniformly over $j\ne r$
proves~\eqref{eq:projection-refresh}. Averaging that identity over an
independently uniform $r$ gives $\E\nu'=\nu$.
\end{proof}

\paragraph{Normalization and projection.}
The following normalization makes makes the average second-moment matrix
close to the identity.

\begin{lemma}[Normalization]\label{lem:entropy-balance}
Let $0<\varepsilon\le1/2$ and let $X_1,\ldots,X_d\in\R^d$ be independent
rows of finite Gaussian relative entropy. They can be replaced by independent
rows, again denoted $X_i$, with the same flip probabilities $\nu_i$, no greater total Gaussian
relative entropy, and second moments $M_i:=\E X_iX_i^\top$ satisfying
\begin{equation}\label{eq:radial-normalization}
 M_i=d\,\E\frac{X_iX_i^\top}{\|X_i\|^2}
 \qquad \mbox{ and } \qquad
 \left\|\frac1d\sum_iM_i-I_d\right\|_{\rm op}\le\varepsilon.
\end{equation} 
\end{lemma}

\begin{proof}
For an invertible deterministic linear map $T$, put
\[
 Y_i(T)=\|G_i\|\frac{TX_i}{\|TX_i\|},\qquad
 F(T)=\sum_i\mathcal K(\mu(Y_i(T))),
 \qquad G_i\sim\gamma_d,
\]
where the $G_i$ are independent of one another and of the rows.
The rows $Y_i(T)$ are independent. Applying a common invertible map
multiplies every determinant by the same nonzero factor, and multiplying
rows by positive scalars maintains the sign. These transformed versions
therefore preserve flip probabilities.

This replacement of lengths cannot increase entropy. Indeed, the map
$(x,g)\mapsto\|g\|x/\|x\|$ sends $\gamma_d\otimes\gamma_d$ to $\gamma_d$.
For any row $Z$ of finite Gaussian relative entropy and an independent
$G\sim\gamma_d$, data processing and
product additivity (Lemma~\ref{lem:entropy-facts}~(\ref{item:data-processing}),
(\ref{item:product-entropy})) give
\[
 \mathcal K\left(\mu\left(\|G\|\frac{Z}{\|Z\|}\right)\right)
 \le D_{\rm KL}\bigl(\mu(Z)\otimes\gamma_d
              \,\Vert\,\gamma_d\otimes\gamma_d\bigr)
 =\mathcal K(\mu(Z)).
\]
Together with~\eqref{eq:linear-entropy-change}, this shows that $F(T)$ is
finite and nonnegative for every invertible $T$, and
$F(I_d)\le\sum_i\mathcal K(\mu(X_i))$.

For fixed $T$, write $Y_i=Y_i(T)$, $M_i=\E Y_iY_i^\top$, and
$\overline M=d^{-1}\sum_iM_i$.
Each new length is independent of its direction and has second moment $d$.
Thus $M_i=d\,\E[Y_iY_i^\top/\|Y_i\|^2]$ and $\tr\overline M=d$.
The matrices $M_i$ are positive definite because the rows have
full-dimensional densities. Set $B=\overline M^{-1/2}$.
For each $i$,
\[
 \frac{BY_i}{\|BY_i\|}=\frac{BTX_i}{\|BTX_i\|}.
\]
Replacing the lengths $\|BY_i\|$ by independent copies of $\|G_i\|$,
independent also of the rows, gives rows distributed as $Y_i(BT)$.
Moreover,
\begin{align*}
 \sum_i\E\|Y_i\|^2
 &=d\,\tr\overline M=d^2,\\
 \sum_i\E\|BY_i\|^2
 &=d\,\tr(B\overline M B^\top)=d^2.
\end{align*}
Hence~\eqref{eq:linear-entropy-change}
and the entropy bound for length replacement give
\[
 F(BT)\le\sum_i\mathcal K(\mu(BY_i))
       =F(T)+\frac d2\log\det\overline M.
\]

Let $\lambda_1,\ldots,\lambda_d$ be the eigenvalues of $\overline M$.
Their sum is $d$, so
\[
 \frac d2\sum_{j=1}^d(\lambda_j-1-\log\lambda_j)
 =-\frac d2\log\det\overline M
 \le F(T)-\inf_{S\text{ invertible}}F(S).
\]
The function $x-1-\log x$ is nonnegative, decreases up to $1$, and
increases thereafter. If $\|\overline M-I_d\|_{\rm op}>\varepsilon$,
the left side is therefore bounded below by a positive constant
depending only on $d$ and $\varepsilon$.
Choose $T$ with $F(T)\le F(I_d)$ and $F(T)$ sufficiently close to its
infimum to exclude this possibility. Such a choice is possible: if a
sufficiently close candidate has $F(T)>F(I_d)$, take $I_d$ instead.
The resulting rows satisfy~\eqref{eq:radial-normalization} and have no
greater total entropy than the original rows.
\end{proof}

The next lemma bounds the entropy removed by projection orthogonal to a single random 
unit vector in terms of the second moment of the unit vector. It does not use logconcavity.

\begin{lemma}
\label{lem:average-projection-entropy}
Let $\mu$ be a probability measure on $\R^d$ with
$\mathcal K(\mu)<\infty$, where $d\ge2$. Let $V$ be a random unit vector
such that $\E VV^\top\succeq\delta I_d$, with $\delta>0$. Then
\begin{equation*}
 \E_V D_{\rm KL}((P_{V^\perp})_\#\mu\Vert\gamma_{V^\perp})
                    \le(1-\delta)\mathcal K(\mu).
\end{equation*}
\end{lemma}

\begin{proof}
We apply Gaussian Brascamp--Lieb to projections of the same fixed measure $\mu$.

Set $P=P_{V^\perp}$ and $k(P)=D_{\rm KL}(P_\#\mu\Vert
P_\#\gamma_d)$, viewing both measures on $\R^d$.
Take independent copies $P_1,P_2,\ldots$ of $P$, and put
\[
 a_N=\left\|\frac1N\sum_{s=1}^N P_s\right\|_{\rm op}.
\]
The average projection is positive semidefinite with trace $d-1$, so
$a_N\ge(d-1)/d>0$.
By the definition of $a_N$,
\[
 \frac{1}{Na_N}\sum_{s=1}^N P_s\preceq I_d.
\]
Thus Lemma~\ref{lem:gaussian-entropic-bl}, with coefficients $1/(Na_N)$, gives
\[
 \frac1N\sum_{s=1}^N k(P_s)\le a_N\mathcal K(\mu).
\]
Data processing (Lemma~\ref{lem:entropy-facts}~(\ref{item:data-processing}))
gives $0\le k(P)\le\mathcal K(\mu)$.
The strong law of large numbers~\cite{Durrett2019}, applied to $k(P_s)$
and to the matrix entries of $P_s$, gives
\[
 \E k(P)\le\|\E P\|_{\rm op}\mathcal K(\mu)
          \le(1-\delta)\mathcal K(\mu),
\]
since $0\preceq\E P=I_d-\E VV^\top\preceq(1-\delta)I_d$.
\end{proof}

\begin{proof}[Proof of Theorem~\ref{thm:sign-entropy}]
Assume $m$ satisfies $m \ge d_0> 10$, where
$d_0=\lceil100\log m\rceil$. 
For any $d$ independent rows in $\R^d$ of finite Gaussian relative entropy,
define 
\[
 k_i=\mathcal K(\mu(X_i)),\qquad K_d=\sum_i k_i,\qquad
 \psi_d=2\nu+ 4d_0\frac{K_d}{d^2}.
\]
We prove, for every $d_0\le d\le m$ and every such collection of rows,
\begin{equation}\label{eq:combined-entropy-sign}
 \psi_d\ge
 \prod_{\ell=d_0+1}^d\left(1+\frac{20}{d_0\ell}\right)^{-1}.
\end{equation}
For the base case, fix $i$ and condition on the other rows, which span
a hyperplane almost surely. The set of values of $X_i$ for which the
determinant is positive is a halfspace through the origin, of Gaussian
measure $1/2$. Let $p=\Prob(f(X)=1\mid X_{-i})$.
Independence leaves the distribution of $X_i$ unchanged, so Pinsker's
inequality (Lemma~\ref{lem:entropy-facts}~(\ref{item:pinsker})) gives
$(p-1/2)^2\le k_i/2$. An independent replacement has conditional
flip probability
\[
 2p(1-p)=\frac12-2\left(p-\frac12\right)^2
 \ge\frac12-k_i.
\]
Averaging over the other rows and then over $i$ gives
\begin{equation*}
 2\nu+\frac{2K_d}{d}\ge1.
\end{equation*}
At $d=d_0$, this implies $\psi_{d_0}\ge1$,
proving~\eqref{eq:combined-entropy-sign} for the empty product.

Now let $d>d_0$. If $K_d\ge d^2/(4d_0)$, then $\psi_d\ge1$ and
\eqref{eq:combined-entropy-sign} holds. Otherwise apply
Lemma~\ref{lem:entropy-balance} with $\varepsilon=1/d_0$.
It preserves $\nu$ and can only decrease $\psi_d$, so it suffices to prove
the bound for the resulting rows. Use the same notation for their
distributions and entropies; we still have $K_d<d^2/(4d_0)$.
Writing $M_i=\E X_iX_i^\top$, their second moments
satisfy~\eqref{eq:radial-normalization} with $\varepsilon=1/d_0$.

In each case below, choose an index $r$ before conditioning on any row
values, condition on $X_r$, and project the other rows onto $X_r^\perp$.
Use fixed orthonormal coordinates in this subspace for the conditional
distributions and their independent replacements.
Let $K_{d-1}$ be the total relative entropy of these conditional
distributions, and write $\nu',\psi_{d-1}'$ for their flip probability
and potential. These quantities depend on $r$ and $X_r$.
Independence leaves the other row distributions unchanged before projection.
Data processing
(Lemma~\ref{lem:entropy-facts}~(\ref{item:data-processing})) gives
\begin{equation}\label{eq:entropy-deletion}
 K_{d-1}\le K_d-k_r.
\end{equation}

First suppose every $k_i\le d/d_0$, and choose $r$ uniformly.
Lemma~\ref{lem:entropy-facts}~(\ref{item:entropy-variational}) applied to
bounded truncations of $\langle u,x\rangle^2/4$, gives for every unit vector $u$
\[
 u^\top M_i u\le4k_i+2\log2,
\]
since $\E_{\gamma_d}e^{\langle u,x\rangle^2/4}=\sqrt2$.
As $d>d_0$, it follows that $M_i\preceq(6d/d_0)I_d$.
For a fixed surviving row $i$, condition on $r\ne i$. The index $r$ is
uniform among the other rows, and $V_r=X_r/\|X_r\|$ is independent of $X_i$.
By~\eqref{eq:radial-normalization},
\begin{align*}
 \E[V_rV_r^\top\mid r\ne i]
 &=\frac1{d(d-1)}\sum_{j\ne i}M_j
 \succeq\frac{1-10/d_0}{d}I_d.
\end{align*}
Apply Lemma~\ref{lem:average-projection-entropy} with
$\delta=(1-10/d_0)/d>0$. Each row survives with probability $(d-1)/d$, so
\[
 \E K_{d-1}\le\frac{d-1}{d}\left(1-\frac{1-10/d_0}{d}\right)K_d.
\]
Dividing by $(d-1)^2$ gives
\[
 \E\!\left[\frac{K_{d-1}}{(d-1)^2}\right]
 \le\left(1+\frac{10}{d_0(d-1)}\right)\frac{K_d}{d^2}
 \le\left(1+\frac{20}{d_0d}\right)\frac{K_d}{d^2}.
\]
Together with $\E\nu'=\nu$ from Lemma~\ref{lem:projection-refresh}, this gives
\begin{equation}\label{eq:potential-step}
 \E\psi_{d-1}'\le\left(1+\frac{20}{d_0d}\right)\psi_d.
\end{equation}

Otherwise choose an index $r$ with $k_r>d/d_0$.
This row carries more than four times the average row entropy.
By~\eqref{eq:entropy-deletion} and $K_d<d^2/(4d_0)$,
\[
 \frac{K_d}{d^2}-\frac{K_{d-1}}{(d-1)^2}
 \ge\frac1{(d-1)^2}
 \left(\frac d{d_0}-\frac{(2d-1)K_d}{d^2}\right)
 \ge\frac1{2d_0d}.
\]
Since $\nu\le1/2$, \eqref{eq:projection-refresh} implies
$2(\E\nu'-\nu)\le1/(d-1)$. Consequently, \eqref{eq:potential-step} holds:
\[
 \E\psi_{d-1}'-\psi_d
 \le\frac1{d-1}-\frac2d\le0.
\]
Conditional on $(r,X_r)$, the projected rows are independent, 
have finite Gaussian relative entropy by~\eqref{eq:entropy-deletion} and full-dimensional densities.
By the induction hypothesis, 
\[
 \prod_{\ell=d_0+1}^{d-1}\left(1+\frac{20}{d_0\ell}\right)^{-1}
 \le\E\psi_{d-1}'
 \le\left(1+\frac{20}{d_0d}\right)\psi_d.
\]
This proves~\eqref{eq:combined-entropy-sign}.
Finally,
\begin{align*}
 \psi_m
 &\ge\prod_{d=d_0+1}^m\left(1+\frac{20}{d_0d}\right)^{-1}\\
 &\ge\exp\left(-\frac{20}{d_0}\sum_{d=d_0+1}^m\frac1d\right)\\
 &\ge\exp\left(-\frac{20\log m}{d_0}\right)
 \ge e^{-0.20}>\frac12.
\end{align*}
If $K_m\le m^2/(16d_0)$, then $4d_0K_m/m^2\le1/4$, so
\[
 2\nu=\psi_m-4d_0\frac{K_m}{m^2}
 \ge\frac12-\frac14=\frac14.
\]
Therefore $\nu\ge1/8$ and the
theorem follows for a sufficiently small $c$.
\end{proof}

\subsection{Large relative entropy implies small part measure}
\label{sec:partition-proof}

\begin{proof}[Proof of Lemma~\ref{thm:intro-main-partition}]
Let $C$ be as in Lemma~\ref{lem:localization-sign}, and let $c_0$ be the
smaller of the constants denoted by $c$ in Lemma~\ref{lem:localization-sign}
and Theorem~\ref{thm:sign-entropy}. Set $c=c_0/C$ and $t=c_0\eta$, so that
$t/C=c\eta$. Lemma~\ref{lem:stable-parts} then gives a sign-stable subfamily
$\mathcal S'$ of of positive measure parts of total measure at least $\eta/2$.

Fix $P=F_1\times\cdots\times F_{n-1}\in\mathcal S'$.
Let $p=\bar\gamma_n^{\otimes(n-1)}(P)$ and for $A \in P$,
let $\widetilde X_i=\sqrt n\,P_{u(P)^\perp}R_i$ in fixed orthonormal coordinates.
By Lemma~\ref{lem:stable-parts}, their flip probability is $\nu_{u(P)}<1/8$,
since the scaling preserves determinant signs.
The map $\sqrt n\,P_{u(P)^\perp}$ sends $\bar\gamma_n$ to
$\gamma_{n-1}$. So, by Lemma~\ref{lem:entropy-facts}~(\ref{item:data-processing})
and~(\ref{item:restriction-entropy}), their total relative entropy satisfies
\[
 \sum_iD_{\rm KL}(\mu(\widetilde X_i)\Vert\gamma_{n-1})
   \le\sum_i\log\frac1{\bar\gamma_n(F_i)}=\log\frac{1}{p}.
\]
This also verifies finiteness of their Gaussian relative entropies.
Theorem~\ref{thm:sign-entropy} now implies
\[
 \log\frac1p>c_0\frac{(n-1)^2}{\log (n-1)}
                 \ge\frac{c_0}{4}\frac{n^2}{\log n}.
\]
Choose $a=c_0/4$. This proves the measure bound and hence the stated cardinality bound.
\end{proof}

The decision-tree consequence follows.

\begin{proof}[Proof of Corollary~\ref{cor:variable-depth-partition}]
Let $\mathcal S'$ be the subfamily from
Lemma~\ref{thm:intro-main-partition}, and set $h=an^2/\log n$.
It has total probability at
least $\eta/2$, and each of its parts has probability at most
$e^{-h}$. For every integer $t\ge0$, at most $(2n+1)^t$ leaves have
depth at most $t$. These leaves contain at most $(2n+1)^t e^{-h}$ of the
probability in $\mathcal S'$. Hence, with $t=\lfloor h/(4\log n)\rfloor$ and $n$ large enough,
\begin{equation*}
                 \Prob(T>t)\ge\frac{\eta}{2}-(2n+1)^t e^{-h} \ge\frac{\eta}{2}-e^{-h/2}\ge\frac{\eta}{4}.
\end{equation*}
Then the assertion follows with $c=a/4$ and implies $\E T\ge a\eta n^2/(16\log^2 n)$.
\end{proof}

\section{Complexity in the oracle model}
\label{sec:oracle-complexity}

We show that, on an event of high Gaussian probability, every separating
point localizes the hidden normal and then apply the product-partition lemma.

\begin{lemma}
\label{lem:gaussian-rounding}
Let $R$ have independent $N(0,1/n)$ entries and set $s=R^{-1}e_n$.
There is a universal constant $C_0\ge2$ such that the event
\begin{equation*}
 \mathcal H=\left\{\max_{1\le i\le n}\|R_i\|\le2,\quad
             \|R^{-1}\|_{\rm op}\le C_0n,\quad
             \frac{\sqrt n}{C_0}\le\|s\|\le C_0\sqrt n\right\}
\end{equation*}
satisfies $\Prob(\mathcal H^c)\le2^{-14}$.
There is a unversal constant $C$ such that on $\mathcal H$,
\begin{equation}\label{eq:hard-rounding}
                B^n\subseteq2K_j(R)\subseteq Cn^{3/2}B^n,
                \qquad j\in\{1,2\},
\end{equation}
\end{lemma}

\begin{proof}
Since $\sqrt n R_i\sim\gamma_n$, Lemma~\ref{lem:gaussian-estimates}
(\ref{item:gaussian-norm-tail}), with $h=3$, and a union bound give
\[
 \Prob\left(\max_{1\le i\le n}\|R_i\|>2\right)
 \le n e^{-n(3-\log4)/2}.
\]
The inverse-norm bound follows from part (\ref{item:gaussian-inverse-tail})
of Lemma~\ref{lem:gaussian-estimates}.
Let $A$ consist of the first $n-1$ rows of $R$. Since $As=0$ and
$\langle R_n,s\rangle=1$,
\[
 s=\frac{v(A)}{\langle R_n,v(A)\rangle},\qquad
 \sqrt n\,\langle R_n,v(A)\rangle\mid A\sim N(0,1).
\]
Thus $\|s\|$ has the distribution of $\sqrt n/|G|$, where $G\sim N(0,1)$.
Using $\Prob(|G|<a)\le\sqrt{2/\pi}\,a$ and a union bound with a sufficiently 
large constant $C_0$ gives
\[
 \Prob(\mathcal H^c)
 \le n e^{-n(3-\log4)/2}
       +\frac{1+\sqrt{2/\pi}}{C_0}+2e^{-C_0^2/2}
 \le\frac1{2^{14}}.
\]

Fix $R\in\mathcal H$. For $\|x\|\le1$, the row bound gives
$|\langle R_i,x/2\rangle|\le1$ for every $i$. Since $b_jL\ge1$,
this shows $x/2\in K_j$ and proves the inner inclusion.
For the outer inclusion, take $x\in K_j$ and set $y=Rx$. Writing
$x=y_ns+R^{-1}(y_1,\ldots,y_{n-1},0)^\top$ gives
\[
 \|x\|\le3L C_0\sqrt n+C_0n^{3/2}=O(n^{3/2}),
\]
since $L$ is a fixed universal multiple of $n$. Taking $C$ sufficiently
large gives~\eqref{eq:hard-rounding}.
\end{proof}

\subsection{Distinguishing the hard pair}
\label{sec:oracle-partition}

For the remainder of Section~\ref{sec:oracle-complexity}, use $s,\mathcal H,C_0$ from
Lemma~\ref{lem:gaussian-rounding}. Let $c$ be the localization constant in
Lemma~\ref{thm:intro-main-partition}, and set
$\varepsilon_0=c/8$ and $L=2C_0^2n/\varepsilon_0$.
In particular, $L>3n$.

\begin{definition}[Hard pair coupling]
\label{def:hard-pair-coupling}
For each fixed $R$, run the same randomized membership-query algorithm on
$K_1(R)$ and $K_2(R)$, with $R_n$ revealed, using the same internal randomness.
Let $\tau$ be the index of the first separating query with $\tau=\infty$ if none occurs.
\end{definition}
Note that the oracles give identical responses prior to the first separating query.
\begin{lemma}[Separating query]
\label{lem:separator-finding}
Let $R$ have independent $N(0,1/n)$ entries. For any randomized
membership-query algorithm with $R_n$ revealed, let $\tau$ be as in
Definition~\ref{def:hard-pair-coupling}.
There is a universal constant $c_0>0$ such that
\begin{equation*}
 t_n=\left\lfloor c_0\frac{n^2}{\log^2 n}\right\rfloor,
 \qquad \Prob(\tau\le t_n)<\frac14.
\end{equation*}
The probability is over the original, unconditioned Gaussian matrix
and the algorithm's internal randomness.
\end{lemma}

\begin{proof}
Fix $R\in\mathcal H$ and a separating point $z$. Set $y=Rz$.
Since $|y_i|\le1$ for $i<n$ and $|y_n|>L$,
\[
 z=y_ns+e,\qquad
 e=R^{-1}(y_1,\ldots,y_{n-1},0)^\top,\qquad
 \|e\|\le C_0n^{3/2}.
\]
The main term $y_ns$ has length at least $L\sqrt n/C_0$.
For nonzero vectors $a,a+e$, we have 
$\|\overrightarrow{a+e}-\overrightarrow a\|\le2\|e\|/\|a\|$.
Since $s$ is parallel to $v(A)$, applying this with $a=y_ns$ gives
\begin{equation}\label{eq:query-localizes}
 \distp(z,v(A))\le\frac{2\|e\|}{|y_n|\|s\|}
                   \le\frac{2C_0^2n}{L}
                   \le\varepsilon_0.
\end{equation}
This holds simultaneously for every separating point, including
one chosen adaptively.

Suppose, for contradiction, that $\Prob(\tau\le t_n)\ge1/4$.
For each fixed choice of the internal randomness and public row, run the algorithm on $K_1$
using the enhanced oracle of Section~\ref{sec:product-partitions}, stopping
at the response to the first separating query or after $t_n$ queries.
This stopping rule is determined by the query and the enhanced response.
Since $R_n$ is known, the inequality
$L<|\langle R_n,z\rangle|\le3L$ can be checked; a NO response identifying
row $n$ certifies that the first $n-1$ constraints hold.
We also treat any transcript after which the algorithm issues no further
query as a terminal leaf, whether or not the algorithm halts.
The stopped tree therefore partitions the first $n-1$ rows into convex product parts.
Each separating leaf $P$ has a fixed query $z(P)\ne0$; set
$u(P)=z(P)/\|z(P)\|$.

To preserve the product structure, we keep the conditional distribution
on each leaf and discard leaves on which $\mathcal H$ fails too often.
Let $p$ denote the conditional probability of $\mathcal H^c$ given the
internal randomness, the public row, and the leaf reached.
For each tree, discard separating leaves with $p>1/1024$.
Averaging over the internal randomness and public row, Markov's inequality
bounds the total probability of the discarded leaves by
\[
 \Prob\!\left(p>\frac1{1024}\right)
 \le1024\,\E p
 =1024\,\Prob(\mathcal H^c)
 \le\frac1{16}.
\]
The remaining separating leaves therefore have average total measure
at least $3/16$. Fix a choice of the randomness and public row for which
this measure is at least $1/8$. By~\eqref{eq:query-localizes}, each remaining
leaf localizes the normal around $u(P)$ with conditional probability
at least $1-1/1024$. Corollary~\ref{cor:variable-depth-partition}, with
$\eta=1/8$, contradicts the depth bound $t_n$ when $c_0$ is smaller than
the constant in that corollary.
\end{proof}

\begin{proof}[Proof of Theorem~\ref{thm:intro-hard-pair-distinguishing}]
Use the coupling of Definition~\ref{def:hard-pair-coupling}, and let $T_j$
be the query count on $K_j$. If $\tau=\infty$, the algorithm can correctly
identify at most one value of $j$, so
\[
 \Prob(\text{correct})\le\frac12+\frac12\Prob(\tau<\infty).
\]
Success probability at least $2/3$ implies $\Prob(\tau<\infty)\ge1/3$.
A first separating query after $t_n$ requires more than $t_n$ queries
in both runs. Lemma~\ref{lem:separator-finding} therefore gives, for
either $j$,
\[
 \Prob(T_j>t_n)\ge\Prob(\tau<\infty)-\Prob(\tau\le t_n)>\frac1{12}.
\]
The expectation lower bound follows immediately.
\end{proof}

\subsection{Optimization, sampling, volume}

By Lemma~\ref{lem:gaussian-rounding}, the bodies $2K_j$ satisfy the required
radius bounds on $\mathcal H$.
A query $x$ to $2K_j$ is simulated by $x/2$ to $K_j$,
with no change in query count.

\begin{proof}[Proof of Theorems~\ref{thm:intro-optimization}
and~\ref{thm:intro-sampling} and
Corollary~\ref{cor:intro-volume}]
For each task, couple the simulated runs as in
Definition~\ref{def:hard-pair-coupling}, with first separating-query index
$\tau$. Let $T$ be the original algorithm's query count on $2K_1$.
Different outputs require a separating query before either run stops.
Fix $R\in\mathcal H$.

For optimization, take $c=R_n/\|R_n\|$. The two optimum values are $h$
and $3h$, where $h=2L/\|R_n\|\ge L>3n$. Under either guarantee, a valid
output for $2K_1$ is at most $h$, whereas a valid output for $2K_2$ is
at least $3h/2>h$.

For volume, $2K_1$ and $2K_2$ have volume ratio three: under the common
linear map $R$, they become boxes differing only by a factor of three in
their last side length. Thus no output interval $[\ell,u]$ with
$u/\ell\le2$ can be valid for both.

For either problem, the union bound gives probability at least $1/3$
that both runs produce valid outputs, which must differ.

For sampling, let $Z_j$ be the output on $2K_j$, and put
$E=\R^n\setminus2K_1$.
The uniform distributions on $2K_1$ and $2K_2$ assign $E$ probabilities
zero and $2/3$, respectively. The total variation guarantees therefore give
\[
 \Prob(Z_1\ne Z_2)
 \ge\Prob(Z_2\in E)-\Prob(Z_1\in E)
 \ge\frac23-\frac2{10}=\frac7{15}>\frac13.
\]
Consequently, for every $R\in\mathcal H$,
\[
                         \Prob(\tau<\infty\mid R)\ge\frac13.
\]
If $t_n<\tau<\infty$, then $T>t_n$.
Applying Lemma~\ref{lem:separator-finding} therefore gives
\[
 \Prob(\mathcal H,\ T>t_n)
 \ge\frac13\Prob(\mathcal H)-\Prob(\tau\le t_n)
 >\frac13\Prob(\mathcal H)-\frac14
 >\frac1{100}\Prob(\mathcal H).
\]
Averaging over $R\in\mathcal H$ supplies an admissible input with
$\Prob(T>t_n)>1/100$. Since $t_n=\Omega(n^2/\log^2 n)$ and
$\E T\ge t_n\Prob(T>t_n)$, the tail and expectation bounds follow.
\end{proof}

\bibliographystyle{abbrv}
\bibliography{acg}

\begin{thebibliography}{10}

\bibitem{Anderson2003}
T.~W. Anderson.
\newblock {\em An Introduction to Multivariate Statistical Analysis}.
\newblock Wiley-Interscience, 3 edition, 2003.

\bibitem{ApplegateK91}
D.~Applegate and R.~Kannan.
\newblock Sampling and integration of near log-concave functions.
\newblock In {\em STOC}, pages 156--163, 1991.

\bibitem{Barany1987}
I.~B{\'a}r{\'a}ny and Z.~F{\"u}redi.
\newblock Computing the volume is difficult.
\newblock {\em Discrete Comput. Geom.}, 2(4):319--326, 1987.

\bibitem{BV04}
D.~Bertsimas and S.~Vempala.
\newblock Solving convex programs by random walks.
\newblock {\em J. ACM}, 51(4):540--556, 2004.

\bibitem{BoucheronLugosiMassart2013}
S.~Boucheron, G.~Lugosi, and P.~Massart.
\newblock {\em Concentration Inequalities: A Nonasymptotic Theory of
  Independence}.
\newblock Oxford University Press, 2013.

\bibitem{BL76}
H.~J. Brascamp and E.~H. Lieb.
\newblock On extensions of the {Brunn-Minkowski} and {Pr{\'e}kopa-Leindler}
  theorems, including inequalities for log concave functions, and with an
  application to the diffusion equation.
\newblock {\em Journal of Functional Analysis}, 22(4):366--389, 1976.

\bibitem{CarlenCorderoErausquin2009}
E.~A. Carlen and D.~Cordero-Erausquin.
\newblock Subadditivity of the entropy and its relation to {Brascamp--Lieb}
  type inequalities.
\newblock {\em Geometric and Functional Analysis}, 19(2):373--405, 2009.

\bibitem{CarlenLieb2008}
E.~A. Carlen and E.~H. Lieb.
\newblock {Brascamp--Lieb} inequalities for non-commutative integration.
\newblock {\em Documenta Mathematica}, 13:553--584, 2008.

\bibitem{CCLW20}
S.~Chakrabarti, A.~M. Childs, T.~Li, and X.~Wu.
\newblock Quantum algorithms and lower bounds for convex optimization.
\newblock {\em Quantum}, 4:221, 2020.

\bibitem{CV18}
B.~Cousins and S.~Vempala.
\newblock {Gaussian} cooling and {$O^*(n^3)$} algorithms for volume and
  {Gaussian} volume.
\newblock {\em SIAM Journal on Computing}, 47(3):1237--1273, 2018.

\bibitem{Din57}
A.~Dinghas.
\newblock {\"U}ber eine {Klasse} superadditiver {Mengenfunktionale} von
  {Brunn-Minkowski-Lusternikschem} {Typus}.
\newblock {\em Math. Zeitschr.}, 68:111--125, 1957.

\bibitem{Durrett2019}
R.~Durrett.
\newblock {\em Probability: Theory and Examples}.
\newblock Cambridge University Press, fifth edition, 2019.

\bibitem{DyerFK89}
M.~E. Dyer, A.~M. Frieze, and R.~Kannan.
\newblock A random polynomial time algorithm for approximating the volume of
  convex bodies.
\newblock In {\em STOC}, pages 375--381, 1989.

\bibitem{Edelman1988}
A.~Edelman.
\newblock Eigenvalues and condition numbers of random matrices.
\newblock {\em SIAM Journal on Matrix Analysis and Applications},
  9(4):543--560, 1988.

\bibitem{E86}
G.~Elekes.
\newblock A geometric inequality and the complexity of computing volume.
\newblock {\em Discrete \& Computational Geometry}, 1:289--292, 1986.

\bibitem{frieze1999log}
A.~Frieze and R.~Kannan.
\newblock Log-{Sobolev} inequalities and sampling from log-concave
  distributions.
\newblock {\em The Annals of Applied Probability}, 9(1):14--26, 1999.

\bibitem{GoyalRademacherVempala2015}
N.~Goyal, L.~Rademacher, and S.~Vempala.
\newblock Query complexity of sampling and small geometric partitions.
\newblock {\em Combinatorics, Probability and Computing}, 24(5):733--753, 2015.

\bibitem{GLS}
M.~Gr{\"o}tschel, L.~Lov{\'a}sz, and A.~Schrijver.
\newblock {\em Geometric Algorithms and Combinatorial Optimization}.
\newblock Springer, 1988.

\bibitem{JLLV26}
H.~Jia, A.~Laddha, Y.~T. Lee, and S.~Vempala.
\newblock Reducing isotropy and volume to {KLS}: Faster rounding and volume
  algorithms.
\newblock {\em Journal of the ACM}, 73(2):1--21, 2026.

\bibitem{KV06}
A.~T. Kalai and S.~Vempala.
\newblock Simulated annealing for convex optimization.
\newblock {\em Math. Oper. Res.}, 31(2):253--266, 2006.

\bibitem{KLS97}
R.~Kannan, L.~Lov\'{a}sz, and M.~Simonovits.
\newblock Random walks and an {$O^*(n^5)$} volume algorithm for convex bodies.
\newblock {\em Random Structures and Algorithms}, 11:1--50, 1997.

\bibitem{Kerger2026}
P.~Kerger.
\newblock Closing the oracle-complexity gap in derivative-free convex
  optimization: A near-quadratic lower bound from exact function values.
\newblock arXiv:2607.13335, 2026.
\newblock \url{https://arxiv.org/abs/2607.13335}.

\bibitem{Kh80}
L.~G. Khachiyan.
\newblock Polynomial algorithms in linear programming.
\newblock {\em USSR Computational Mathematics and Mathematical Physics},
  20:53--72, 1980.

\bibitem{kook2025sampling}
Y.~Kook and S.~S. Vempala.
\newblock Sampling and integration of logconcave functions by algorithmic
  diffusion.
\newblock In {\em Proceedings of the 57th Annual ACM Symposium on Theory of
  Computing}, pages 924--932, 2025.

\bibitem{LSV19}
Y.~T. Lee, A.~Sidford, and S.~S. Vempala.
\newblock Efficient convex optimization with oracles.
\newblock In {\em Building Bridges II}, Bolyai Society Mathematical Studies,
  pages 317--335. Springer Berlin Heidelberg, 2019.

\bibitem{lee2015faster}
Y.~T. Lee, A.~Sidford, and S.~C.-w. Wong.
\newblock A faster cutting plane method and its implications for combinatorial
  and convex optimization.
\newblock In {\em Foundations of Computer Science (FOCS), 2015 IEEE 56th Annual
  Symposium on}, pages 1049--1065. IEEE, 2015.

\bibitem{Lei72}
L.~Leindler.
\newblock On a certain converse of {H}\"older's inequality {II}.
\newblock {\em Acta Sci. Math. Szeged}, 33:217--223, 1972.

\bibitem{L90}
L.~Lov\'{a}sz.
\newblock How to compute the volume?
\newblock {\em Jber. d. Dt. Math.-Verein, Jubil\"aumstagung 1990}, pages
  138--151, 1990.

\bibitem{LS93}
L.~Lov\'asz and M.~Simonovits.
\newblock Random walks in a convex body and an improved volume algorithm.
\newblock In {\em Random Structures and Alg.}, volume~4, pages 359--412, 1993.

\bibitem{LV06}
L.~Lov\'{a}sz and S.~Vempala.
\newblock Fast algorithms for logconcave functions: sampling, rounding,
  integration and optimization.
\newblock In {\em FOCS}, pages 57--68, 2006.

\bibitem{LV07}
L.~Lov{\'a}sz and S.~Vempala.
\newblock The geometry of logconcave functions and sampling algorithms.
\newblock {\em Random Struct. Algorithms}, 30(3):307--358, 2007.

\bibitem{PolyanskiyWu2025}
Y.~Polyanskiy and Y.~Wu.
\newblock {\em Information Theory: From Coding to Learning}.
\newblock Cambridge University Press, 2025.

\bibitem{Pr73b}
A.~Pr{\'e}kopa.
\newblock Logarithmic concave measures with application to stochastic
  programming.
\newblock {\em Acta Sci. Math. Szeged}, 32:301--316, 1971.

\bibitem{Prekopa1973}
A.~Pr{\'e}kopa.
\newblock On logarithmic concave measures and functions.
\newblock {\em Acta Scientiarum Mathematicarum}, 34:335--343, 1973.

\bibitem{Protasov1996}
V.~Y. Protasov.
\newblock Algorithms for approximate calculation of the minimum of a convex
  function from its values.
\newblock {\em Mathematical Notes}, 59(1):69--74, 1996.

\bibitem{RademacherShao2008}
L.~Rademacher and X.~Shao.
\newblock Minimal partitioning into product sets.
\newblock Manuscript, 2008.
\newblock \url{https://www.math.ucdavis.edu/~lrademac/partition.pdf}.

\bibitem{RV08}
L.~Rademacher and S.~Vempala.
\newblock Dispersion of mass and the complexity of randomized geometric
  algorithms.
\newblock {\em Advances in Mathematics}, 219(3):1037--1069, 2008.

\bibitem{Shor77}
N.~Z. Shor.
\newblock Cut-off method with space extension in convex programming problems.
\newblock {\em Cybernetics}, 13(1):94--96, 1977.
\newblock Russian original in {Kibernetika}, no. 1, pp. 94--95, 1977.

\bibitem{Vai}
P.~M. Vaidya.
\newblock A new algorithm for minimizing convex functions over convex sets.
\newblock {\em Mathematical Programming}, 73:291--341, 1996.

\bibitem{AGGW20}
J.~van Apeldoorn, A.~Gily{\'e}n, S.~Gribling, and R.~de~Wolf.
\newblock Convex optimization using quantum oracles.
\newblock {\em Quantum}, 4:220, 2020.

\bibitem{Yao1977}
A.~C.-C. Yao.
\newblock Probabilistic computations: Toward a unified measure of complexity.
\newblock In {\em Proceedings of the 18th Annual Symposium on Foundations of
  Computer Science (FOCS)}, pages 222--227. IEEE Computer Society, 1977.

\bibitem{YN76}
D.~B. Yudin and A.~S. Nemirovski.
\newblock Informational complexity and efficient methods for the solution of
  convex extremal problems (in {Russia}n).
\newblock {\em Ekonomika i Matematicheskie Metody}, 12(2):357--369, 1976.
\newblock English translation in {Matekon}, 13(3):25--45, 1977.

\end{thebibliography}

\end{document}